\documentclass[conference]{IEEEtran}

\ifCLASSINFOpdf
  \usepackage[pdftex]{graphicx}
  \graphicspath{{./}}
  \DeclareGraphicsExtensions{.pdf,.jpeg,.png}
\else
  \usepackage[dvips]{graphicx}
  \graphicspath{{./}}
  \DeclareGraphicsExtensions{.eps}
\fi

\usepackage[cmex10]{amsmath}
\usepackage{amssymb,amsthm,amsfonts}
\usepackage{graphicx}
\usepackage{cite}
\usepackage{mathabx}
\usepackage{accents} % to put a underline bar in the definition of the vect. Should be loaded after mathabx 
\usepackage{flushend}

\newtheorem{theorem}{Theorem}
\newtheorem{lemma}[theorem]{Lemma}
\theoremstyle{definition}

\theoremstyle{remark}
\newtheorem{remark}{Remark}

\DeclareMathOperator{\sinc}{sinc}
\DeclareMathOperator{\vol}{vol}
\providecommand{\norm}[1]{\left\lVert#1\right\rVert}
\newcommand{\nn}{\nonumber}

\newcommand{\vect}[1]{\underaccent{\bar}{#1}}
\newcommand{\const}[1]{{\mathcal{#1}}}
\newcommand{\Reals}{\mathbb{R}}
\newcommand{\Complex}{\mathbb{C}}
\newcommand{\E}{\mathsf{E}}
\newcommand{\ie}{\emph{i.e.}}
\newcommand{\eg}{\emph{e.g.}}
\newcommand{\der}{\mathrm{d}}
\newcommand{\cc}{\mathrm{c.c.}}

\newcommand{\eqdef}{\stackrel{\Delta}{=}}
\newcommand{\snr}{\text{SNR}}
\newcommand{\bw}{\mathcal{B}}

\begin{document}

\thispagestyle{empty}

\title{Upper Bound on the Capacity of the Nonlinear
  Schr\"odinger Channel}

\author{\IEEEauthorblockN{Mansoor~I.~Yousefi\IEEEauthorrefmark{1}\IEEEauthorrefmark{2}, 
Gerhard~Kramer\authorrefmark{1}\IEEEauthorrefmark{2} and
  Frank~R.~Kschischang\IEEEauthorrefmark{2}\IEEEauthorrefmark{3}}
\IEEEauthorblockA{\IEEEauthorrefmark{1}
Institute for Communications Engineering, Technical University of Munich, Germany\\
}
\IEEEauthorblockA{\IEEEauthorrefmark{2}
Institute for Advanced Study, Technical University of Munich, Germany\\
}
\IEEEauthorblockA{\IEEEauthorrefmark{3}
Edward S. Rogers Sr. Dept. of Electrical \&
Computer Engineering, University of Toronto, Canada
}
}

\date{}

\maketitle

\IEEEpeerreviewmaketitle
\begin{abstract}
It is shown that the capacity of the channel modeled by
(a discretized version of) the stochastic
nonlinear Schr\"odinger (NLS) equation is upper-bounded by $\log(1+\snr)$
with $\snr=\const P_0/\sigma^2(z)$, where $\const P_0$ is the average input
signal power and $\sigma^2(z)$ is the total noise power up to distance $z$.
The result is a consequence of the fact that the deterministic NLS equation is
a Hamiltonian energy-preserving dynamical system.
\end{abstract}

\section{Introduction}

Half a century after the introduction of the optical fiber, the problem of
determining its capacity remains open.  This holds even for the single-user
point-to-point channel subject to a power and bandwidth constraint. There
is also a lack of general upper bounds, as well as lower bounds in the
high-power regime. The asymptotic capacity when
power $\const P\rightarrow\infty$ is also unknown.

Numerical simulations of the optical fiber channel with additive white
Gaussian noise (AWGN) seem to indicate that the data rates that can be
achieved using current methods are below $\log(1+\snr)$, the capacity of an
AWGN channel with signal-to-noise ratio $\snr$. In this paper, we prove
this conjecture, namely, we show that
\begin{equation}
C\leq \log(1+\snr),
\label{eq:upper-bound}  
\end{equation}
where $\snr(z)\eqdef\const P_0/\sigma^2(z)$, in which $\const P_0$ is the
average input signal power and $\sigma^2(z)$ is the total noise power up to
the distance $z$. Here $C$ is the capacity of the point-to-point channel
per complex degree-of-freedom. 

Motivated by recent developments suggesting that the nonlinearity can
be constructively taken into account in the design of communication
schemes to potentially address the capacity bottleneck problem in optical
fiber \cite{yousefi2012nft1,yousefi2012nft2,yousefi2012nft3}, it has been
speculated that data rates above $\log(1+\snr)$ may even be achievable.
While the nonlinearity can be exploited, as for instance in
\cite{yousefi2012nft1,yousefi2012nft2,yousefi2012nft3,wahls2014fnft,zhang2014sps,
prilepsky2013nsm}, the upper bound \eqref{eq:upper-bound} shows that it
does not offer any gain in capacity relative to the linear channel.  All
one can hope for is to embrace nonlinearity in the communication design so
that it does not penalize the capacity at high powers. This is
expected in the (closed) conservative system \eqref{eq:nls}, which does not
include any gain (amplification) mechanism.

Throughout this paper, lower and upper case letters represent,
respectively, deterministic and random variables. Row vectors are denoted
by underline, \eg, $\vect Q^n\eqdef(Q_1,\cdots, Q_n)$. As usual,
$\Reals$, resp.~$\Complex$, denotes the set of real, resp.~complex, numbers.
The imaginary unit is denoted by $j=\sqrt{-1}$. 

\section{Continuous-time Channel Model and Its Discretization}

Let $Q(t,z): \Reals\times \Reals^+\mapsto \Complex$ be a function of time
$t$ and space $z$. Signal propagation in optical fiber is described by the
stochastic nonlinear Schr\"odinger (NLS) equation \cite[Eq.~3]{yousefi2012nft1} 
\begin{IEEEeqnarray}{rCl}
j\partial_z Q=\partial_{tt} Q+2|Q|^2Q+W(t,z).
\label{eq:nls}
\end{IEEEeqnarray}
Here $W(t,z)$ is space-time white circularly symmetric complex Gaussian noise  with constant power
spectral density $\sigma_0^2$ and bandlimited to $[-\bw/2, \bw/2]$, \ie,
\begin{IEEEeqnarray*}{rCl}
  \E \left(W(t,z)W^*(t',z')\right)=\sigma_0^2\delta_{\bw}(t-t')\delta(z-z'),
\end{IEEEeqnarray*}
where $\delta_{\bw}(x)\eqdef\bw\sinc(\bw x)$, $\sinc(x)\eqdef\sin(\pi x)/(\pi x)$, and $\delta(x)$ is the Dirac delta function.
The transmitted signal power is limited so that
\begin{IEEEeqnarray}{rCl}
\lim\limits_{\mathcal T\rightarrow\infty}\E\frac{1}{\mathcal T}\int\limits_{-\mathcal T/2}^{\mathcal T/2}|Q(t,0)|^2\der t\leq \const{P}_0.
\label{eq:power-cont}
\end{IEEEeqnarray}

We discretize the continuous-time model \eqref{eq:nls} by considering the
partial differential equation (PDE) \eqref{eq:nls} with periodic boundary
conditions
\[
Q(t+T,z)=Q(t,z),\quad \forall\: t,z,
\]
where $T$ is the signal period.  Substituting the two-dimensional Fourier
series (see \cite[Sections III and V]{yousefi2014psd})
\begin{IEEEeqnarray*}{rCl}
    Q(t,z)= \sum\limits_{k=-\infty}^{\infty}Q_k(z)e^{j(k\omega_0
      t+k^2\omega_0^2 z)},
\end{IEEEeqnarray*}
into the NLS equation \eqref{eq:nls}, we obtain 
\begin{IEEEeqnarray}{rCl}
  \partial_z Q_k(z)&=&-2j\sum\limits_{lmn}e^{j\Omega_{lmnk}
    z}Q_l(z)Q_m(z)Q_n^*(z)\delta_{lmnk}\nn\\
&&+\:W_k(z),
\label{eq:nlsf}
\end{IEEEeqnarray}
where $\delta_{lmnk}\eqdef\delta[l+m-n-k]$, $\delta[k]$ is the
Kronecker delta function, and
\begin{IEEEeqnarray*}{rCl}
  \Omega_{lmnk}\eqdef\omega_0^2(l^2+m^2-n^2-k^2),\quad \omega_0\eqdef 2\pi/T.
\end{IEEEeqnarray*}
We assume that $T\rightarrow\infty$ so that the discrete model \eqref{eq:nlsf} captures the infinitely many 
signal degrees-of-freedom in the
continuous model \eqref{eq:nls} in a one-to-one manner. As a result, $W_k$ are uncorrelated circularly 
symmetric complex Gaussian random variables, with 
\[
  \E
  \left(W_k(z)W_{k'}^*(z')\right)=f_0\sigma^2_0\delta[k-k']\delta(z-z'),\quad
  f_0\eqdef\omega_0/2\pi.
\]

The coupled stochastic ordinary differential equation  (ODE) system
\eqref{eq:nlsf} defines a discrete vector communication channel in the
frequency domain $\vect{Q}^n(0)\mapsto \vect{Q}^m(z)$. For notational
convenience, we limit to positive frequencies so that vector indices start
from one. We denote the action of the stochastic ODE system \eqref{eq:nlsf} on input $\vect{Q}^n(0)$ by $S_z$, \ie,
$\vect{Q}^m(z)=S_z(\vect{Q}^n(0))$. We denote the action of the deterministic
(noiseless)
system (where $\vect W^n=0$) on input $\vect{Q}^n(0)$
by $T_z$, \ie,
$\vect{Q}^m(z)=T_z(\vect{Q}^n(0))$.
The power constraint \eqref{eq:power-cont} is discretized to $\const P(0)\leq \const P_0$, where 
\begin{equation}
\const P(z)\eqdef\sum\limits_{k=1}^n \E|Q_k(z)|^2. 
\label{eq:power}
\end{equation}
In this paper, we assume $n=m$ and study the
capacity of the discretized channel $S_z$, instead of the original
continuous-time channel \eqref{eq:nls}. See Remark~\ref{rem:spectral-broadening} for the case $n\neq m$.

The upper bound \eqref{eq:upper-bound} on the capacity of $S_z$ is
obtained as follows. The transformation $T_z$ is energy-preserving,
implying that the output power in $S_z$ is $\const{P}(0)+\sigma^2(z)$,
$\sigma^2(z)\eqdef\bw \sigma_0^2 z$. Consequently, the output (differential)
entropy rate is upper-bounded, from the maximum entropy theorem, by
$C_n+\log(\const{P}_0+\sigma^2(z))$, $C_n\eqdef\log(\pi e/n)$. For the
conditional entropy, note that noise is added continuously along the link.
The entropy power inequality (EPI) implies that the conditional entropy
rate is not less than the (overall) noise entropy rate $C_n+\log(\sigma^2(z))$. Combining these two
results, $C\leq \log(1+\snr)$. In what follows, we establish these two
steps.

The use of the EPI in bounding the conditional entropy rate is an important step
in our proof. It is therefore worth elaborating on the EPI briefly, to see
why entropy should increase at least by a constant amount at each
point that noise
is added along the link. In Appendix~\ref{app:epi}, we briefly review this
interesting inequality.

\section{Upper Bound}

\subsection{Upper Bound on the Output Entropy}

\begin{lemma}[Monotonicity of the Power in $S_z$] Let $\bw$ be the common
signal and noise (passband) bandwidth from input to output. The output
average power in $S_z$ is 
\begin{equation}
\const{P}(z)=\const{P}(0)+\sigma^2(z).  
\label{eq:pz}
\end{equation}
\label{lemm:power}
\end{lemma}

\begin{proof}
Since the signal and noise are commonly bandlimited to $\bw$, $Q_k$ and
$W_k$ are supported in $1\leq k\leq n$  for all $z$, $n=\bw/f_0$. Taking
the derivative with respect to $z$ in \eqref{eq:power} and using \eqref{eq:nlsf}, we obtain
\begin{IEEEeqnarray}{rCl}
\frac{\der \const P(z)}{\der z}
&=&4\Im\Bigl(\sum\limits_{lmnk}\E Q_lQ_mQ_n^*Q_k^*e^{j\Omega_{lmnk}z}\Bigr)
   \nn\\
&& +\sum\limits_{k=1}^n  \E(Q_k^*W_k+Q_kW_k^*)\nn\\
&=& \sum\limits_{k=1}^n \E\Bigl(Q_k^*(z)W_k(z)+Q_k(z)W_k^*(z)\Bigr),
\label{eq:dE-dz}  
\end{IEEEeqnarray}
where we used the fact that the nonlinear term is real-valued,
since $\Omega_{lmnk}=-\Omega_{nklm}$. We now integrate \eqref{eq:dE-dz} in
distance. From \eqref{eq:nlsf}, $Q_k(z)$ contains a term depending on
$W_k(l)$, $l<z$, and a Brownian motion term $B_k(z)=\int_0^z W_k(l)\der l$.
The first term is independent of $W_k(z)$; from the second term we get
\begin{IEEEeqnarray*}{rCl}
\E\bigl(\int\limits_0^z \left(Q_k^*(l)W_k(l)+\cc\right)\der l\bigr)&=&\E\bigl(\int\limits_0^z \left(  B^*_k(l)\der  B_k(l)+\cc\right)\bigr)\\
&=&\E|B_k(z)|^2 \\
&=&f_0 \sigma_0^2z,
\end{IEEEeqnarray*}
where $\cc$ stands for complex conjugate. Summing over $1\leq k\leq n$, we
obtain \eqref{eq:pz}.
\end{proof}

Using Lemma~\ref{lemm:power}, the output entropy rate can be upper bounded as
follows:
\begin{IEEEeqnarray}{rCl}
\frac{1}{n}h(\vect Q^n(z))& \overset{(a)}{\leq} & \frac{1}{n}\log \left((\pi e)^n\det K(z)\right)\nn\\
&=&\log\pi e+\frac{1}{n}\log\left(\det K(z)\right)\nn\\
&\overset{(b)}{\leq}& \log\pi e+\frac{1}{n}\sum\limits_{k=1}^n\log\left(K_{kk}(z)\right)\nn\\
&\overset{(c)}{\leq}& \log\pi e+\frac{1}{n}\sum\limits_{k=1}^n\log\left(\E|Q_k(z)|^2\right)\nn\\
&\overset{(d)}{\leq}& \log\pi e+\log\Bigl(\frac{1}{n}\E\sum\limits_{k=1}^n
  |Q_k(z)|^2\Bigr)\nn\\
&=& \log\pi e+\log(\const P(z)/n)\nn\\
&\overset{(e)}{\leq}& C_n+\log\left(\const P_0+\sigma^2(z)\right),
\label{eq:output-entropy}
\end{IEEEeqnarray}
where $K(z)>0$ is the covariance matrix of $\vect Q^n(z)$ with entries $K_{kl}(z)$.  Step $(a)$ is due
to the maximum entropy theorem. Step $(b)$ follows from Hadamard's
inequality. For step $(c)$, note that in \eqref{eq:power-cont}, power was defined as average energy in 
time interval $\mathcal T$ divided by $\mathcal T$. As a result, a non-zero
 constant
signal has non-zero power. In the covariance matrix, in contrast to \eqref{eq:power-cont} and \eqref{eq:power}, the 
mean of the random variable is 
subtracted as $K_{kk}(z)=\E|Q_k(z)|^2-\left|\E(Q_k(z))\right|^2$. Unlike \eqref{eq:power}, the mean term 
$\sum_{k=1}^n\left|\E Q_k(z)\right|^2$ is not preserved 
in the noise-free channel. Furthermore, a zero-mean signal at the input may 
not have zero mean at $z>0$. Nevertheless, step $(c)$ holds since $K_{kk}(z)\leq \E|Q_k(z)|^2$.
Steps (d) and (e) follow, respectively, from the concavity of the $\log$ function and \eqref{eq:pz}. In steps $(b)$, $(c)$ 
and $(e)$, we also used the fact that $\log$ is an increasing function.

\subsection{Lower Bound on the Conditional Entropy}

\begin{lemma}[Volume Preservation in $T_z$]
Let $\Omega=(\ell^2,\mathcal E, \mu)$ be a measure space, where
$\ell^2\eqdef \left\{ \vect q^n \:|\: \sum|q_k|^2<\infty\right\}$  and
\[
\mu(A)=\vol(A)=\int\limits_{A}\left(\prod\limits_{k=1}^{n}\der
  q_k\der q^*_k\right),\quad \forall A\in\mathcal E,
\]
is the Lebesgue measure. Transformation $T_z$, as a dynamical system on
$\Omega$, is measure-preserving. That is to say
\[
 \mu(T_z^{-1}(A))= \mu(A),\quad \forall A\in\mathcal E.
\]
\label{lemm:volume}
\end{lemma}
\begin{proof}
We note that, when $W_k=0$, the ODE system \eqref{eq:nlsf} is Hamiltonian,
\ie, it permits an alternative formulation
\begin{equation}
\dot x_k=\frac{\partial H}{\partial y_k},\quad \dot y_k=-\frac{\partial H}{\partial x_k},\quad k=1,\cdots,n,
\label{eq:hamilton}
\end{equation}
where dot represents $\partial_z$ , $(x_k,y_k)=(q_k,q^*_k)$ and the
Hamiltonian function $H$ is given by
\begin{IEEEeqnarray*}{rCl}
  H(\vect x^n,\vect y^n)&=&j
\sum\limits_{k=1}^{n}\omega_0^2k^2 x_k y_k\\
&&-\:j\sum\limits_{abcd=1}^{n}x_ax_by_cy_de^{j\Omega_{abcd}z}\delta_{abcd}
.
\end{IEEEeqnarray*} 
Liouville's theorem asserts that Hamiltonian systems preserve the
Lebesgue measure \cite{viarnoldbook}. This is indeed easy to see. Let
$\der\mu=\prod_{k=1}^n \der x_k \der y_k$. Then
\begin{IEEEeqnarray*}{rCl}
  \der\dot\mu&=&\left(\der
  \dot x_1 \der y_1+ \der x_1\der  \dot y_1 \right)\prod\limits_{k=2}^n \der x_k \der y_k
+\cdots\\
&=&0,
\end{IEEEeqnarray*}
where  we substituted \eqref{eq:hamilton}. It follows that $T_z$ is a
volume-preserving transformation (in the sense of ergodic theory
\cite{walters2000ergodic}).
\end{proof}

\begin{lemma}[Entropy Preservation in $T_z$]
The flow of $T_z$ is entropy-preserving, \ie,
$h(T_z^{-1}(\vect Q^n))=h(\vect Q^n)$.
\label{lemm:entropy-preservation}
\end{lemma}

\begin{proof}
From Lemma~\ref{lemm:volume},  $T_z$ is a measure-preserving
transformation; therefore it has unit (determinant) Jacobian, $\det J=1$, where 
$J$ is the $\Reals^{2n\times 2n}$ Jacobian matrix. Since $T_z$ is
also invertible
\begin{IEEEeqnarray*}{rCl}
h(T_z^{-1} (\vect Q^n))=h(\vect Q^n)-\E\log|\det J|=h(\vect Q^n).
\end{IEEEeqnarray*}
Note that with $J$ as a $\Complex^{n\times n}$ matrix, there 
would be a factor $2$ in front of the $\log$.
\end{proof}

In the example of the NLS channel \eqref{eq:nls}, the dispersion and
nonlinear parts are separable and can be solved in simple forms. In
such examples, it might be possible to directly check that the flow of the
equation has unit Jacobian.  Note that the dispersion operator, being a
unitary transformation, has unit Jacobian. One can also verify that the
nonlinear part of the NLS equation \eqref{eq:nls}  has unit Jacobian too.
Consider 
\begin{equation}
Y=X\exp(j f(|X|)),\quad X,Y\in\Complex,
\label{eq:nl-part}  
\end{equation}
for any differentiable function $f(X)$. In \eqref{eq:nls}, $f(X)=zX^2$,
$X=Q(t,0)$ and $Y=Q(t,z)$. Linearizing at $X=0$, $\der Y=\der X$. More
formally, in polar coordinates
\begin{IEEEeqnarray*}{rCl}
  R_Y =R_X,\quad \Phi_Y =\Phi_X +f(R_X),
\end{IEEEeqnarray*}
where $(R_X,\Phi_X)$ and $(R_Y, \Phi_Y)$ are coordinates of $X$ and $Y$,
respectively. Clearly $\det J=1$, which can be seen is the same in the Cartesian coordinates because $|Y|=|X|$.
Since the transformation from the NLS equation \eqref{eq:nls} in the time
domain to the ODE system \eqref{eq:nlsf} in the discrete frequency domain
is also unitary and unit Jacobian, $T_z$ has unit Jacobian.

Finally, it is also possible to check that $T_z$ is entropy-preserving
using the elementary properties of the entropy. It is obvious that the
dispersion operator is entropy-preserving. In the continuous model
\eqref{eq:nls}, the nonlinear transformation in each time sample is given
by \eqref{eq:nl-part}. Using the chain rule for entropy 
\begin{IEEEeqnarray*}{rCl}
h(R_Y,\Phi_Y)&=& h(R_Y)+h\left(\Phi_Y\bigl|R_Y\right)\nn\\
&=&h(R_X)+h\Bigl(\Phi_X+f(R_X)\bigl|R_X\Bigr)\nn\\
&=&h(R_X)+h\left(\Phi_X \bigl|R_X\right)\nn\\
&=&h(R_X,\Phi_X).
%\label{eq:hy=hx}
\end{IEEEeqnarray*}
Note that the entropy of a complex random variable is defined as the
joint entropy of the real and imaginary parts. Changing variables to the Cartesian coordinate system shifts the entropy 
by  $\E\log |\det J|=\E\log R_Y=\E\log R_X$. Thus $h(R_Y\exp(j\Phi_Y))=h(R_X\exp(j\Phi_X))$.
The result also holds for the vector version of
\eqref{eq:nl-part} as well. 
Because the Fourier transform is also entropy-preserving, so is 
$T_z$.

The last two approaches, however, depend on details of the example at hand.
For some equations the nonlinear part is not an additive term to dispersion, and even 
if it is, it may not be simply solvable like
\eqref{eq:nl-part}. For instance, the nonlinear part of the
Korteweg-de~Vries (KdV) equation is Burgers' equation, which is not easily solvable
as \eqref{eq:nl-part}, so as to examine entropy preservation directly.
However, it is quite easy to show that the KdV equation, and indeed a large
number of evolution equations, are Hamiltonian.

\begin{lemma}[Monotonocity of the Entropy in $S_z$]
The conditional entropy rate in $S_z$ is lower-bounded by the noise entropy
rate, \ie,
\[
\frac{1}{n}h(S_z(\vect Q^n(0))|\vect Q^n(0))\geq C_n+\log\sigma^2(z) .
\]
\end{lemma}

\begin{proof}
In a small interval $\Delta z$ in \eqref{eq:nlsf}
\begin{IEEEeqnarray}{rCl}
S_{z+\Delta z}(\vect Q^n(z))=T_{\Delta z}(S_z(\vect Q^n(z)))+\vect W^n(z)\sqrt{\Delta z}.
\label{eq:Sz+dz}
\end{IEEEeqnarray}
The two terms in the right hand side of \eqref{eq:Sz+dz} are independent. Applying the EPI \eqref{eq:epi},
\begin{IEEEeqnarray*}{rCl}
2^{\frac{1}{n}h(S_{z+\Delta z}(\vect Q^n))} &\geq&
2^{\frac{1}{n}h(T_{\Delta z}(S_z(\vect
  Q^n)))}+2^{\frac{1}{n}h(\vect W^n(z)\sqrt{\Delta z})}\\
&=& 
2^{\frac{1}{n}h(S_z(\vect
  Q^n))}+2^{C_n}\sigma^2(\Delta z),
\end{IEEEeqnarray*}
where the last step follows because, from
Lemma~\ref{lemm:entropy-preservation}, $T_z$ is entropy-preserving and
$\vect W^n$ is Gaussian.  Given $\vect Q^n(0)=\vect q^n(0)$, we integrate
in $z$ to obtain
\begin{IEEEeqnarray*}{rCl}
2^{\frac{1}{n}h(S_{z}(\vect q^n(0)))} 
&\geq& 
2^{\frac{1}{n}h(\vect
  q^n(0))}+2^{C_n}\sigma^2(z)\\
&=&
2^{C_n} \sigma^2(z).
\end{IEEEeqnarray*}
It follows that
\begin{IEEEeqnarray}{rCl}
  \frac{1}{n}h(\vect Q^n(z)|\vect Q^n(0))\geq C_n+\log(\sigma^2(z)).
\qedhere
\label{eq:cond-entropy}
\end{IEEEeqnarray}
\end{proof}

Combining \eqref{eq:output-entropy} and \eqref{eq:cond-entropy}, we bound
the mutual information
\[
  \frac{1}{n}I(\vect Q^n(0); \vect Q^n(z))\leq
  \log\left(1+\frac{\const P_0}{\sigma^2(z)}\right).
\]
Noting that the right hand side is independent of the input distribution,
we obtain the upper bound \eqref{eq:upper-bound}.

\begin{remark}[Spectral Efficiency (SE) in the Case $n\neq m$]
In this paper, we did not introduce filters into the model.
Any potential filtering at the receiver (possibly due to spectral
broadening) can only decrease the mutual information (by the chain rule). Furthermore, let $B(z)$ be 
the bandwidth at distance $z$, according to a certain definition. Since $B(0)\leq \max_z B(z)$, 
normalizing by $\max_z B(z)$ would only decrease the SE relative to the 
linear dispersive channel (where $B(z)=B(0)$). 
In summary, nonlinearity is entropy-preserving and the effect of
its spectral broadening does not increase data rate or the SE. The upper bound 
\eqref{eq:upper-bound} on the SE holds if $n\neq m$. 

Throughout the paper, we assumed that noise bandwidth is larger
than the signal bandwidth. Otherwise, capacity can be (nearly) unbounded by exploiting
the (nearly) noise-free frequency band. 
\label{rem:spectral-broadening}
\qed
\end{remark}

The upper bound \eqref{eq:upper-bound} is indeed simple. In this paper, we discussed 
it in the context of a general Hamiltonian channel with continuous evolution. In particular, 
it also holds for a discrete concatenation of energy- and entropy-preserving systems with additive white Gaussian 
noise.

A different account of the upper bound \eqref{eq:upper-bound} is given in \cite{kramer2015ubc} using the
split-step Fourier method. 

\section{Conclusion}

It is shown that the capacity of the point-to-point optical fiber channel,
modeled via the stochastic nonlinear Schr\"odinger equation
(\ref{eq:nls}), and subject to a power and bandwidth constraint, is
upper-bounded by $\log(1+\snr)$.

\section*{Acknowledgment}
All authors acknowledge the support of the Institute for Advanced Study at 
the Technical University of Munich, funded by the German Excellence
Initiative. M. I. Yousefi and G. Kramer were also supported by an Alexander von Humboldt 
Professorship, endowed by the German
Federal Ministry of Education and Research.

\appendices
\section{The Entropy Power Inequality}
\label{app:epi}

\begin{lemma}[Entropy Power Inequality]
Let $X, Y\in\Reals^n$ be independent random variables. Define the entropy
power of a random variable $X\in\Reals^n$  as
\begin{IEEEeqnarray}{rCl}
\sigma_e^2(X)=\frac{1}{2\pi e}2^{\frac{2}{n}h(X)}.
\label{eq:volume-epi}  
\end{IEEEeqnarray}
Then 
\begin{IEEEeqnarray}{rCl}
\sigma_e^2(X+Y)\geq \sigma_e^2(X)+\sigma_e^2(Y).  
\label{eq:epi}
\end{IEEEeqnarray}
Equality holds if and only if $X$ and $Y$ are Gaussian with
proportional covariance matrices.
\end{lemma}
\begin{proof}
By now there are many proofs of the EPI. A simple proof is given in \cite[Section 17.8]{cover2006}. It can be explained as
follows. 

Consider $n=1$. We are looking for an inequality involving the convolution
$f_{X}(x)\convolution f_Y(y)$. The well-known Young's inequality for
$f_X(x)\in L^p(\Reals)$ and $f_Y(y)\in L^q(\Reals)$ states 
\begin{IEEEeqnarray}{rCl}
 \norm{f_{X}(x)\convolution f_Y(y)}_a\leq C
\norm{f_{X}(x)}_p\norm{f_{Y}(y)}_q,
\label{eq:young} 
\end{IEEEeqnarray}
where $1/p+1/q=1/a+1$ ($p,q,a\geq 1$), and
$C=\sqrt{C_pC_q/C_a}$, $C_x=x^{\frac{1}{x}}/{x'}^{\frac{1}{x'}}$, where
$x'$ is conjugate to $x$, \ie, $1/x+1/x'=1$. When $p,q\neq 1$, the equality
holds if and only if $f_X(x)$ and $f_Y(y)$ are Gaussian. On the other hand, entropy and norm of a probability
density $f_X(x)$ are related via
$h(X)=-\partial_a\log\norm{f_X(x)}_a^a$ at $a=1$. However differentiating
both sides of an inequality does not preserve the sense
of the inequality. Nevertheless, using
L'H\^opital's rule we can convert differentiation to a limit
\[
  h(X)=\lim\limits_{a\downarrow 1}\frac{1}{1-a}\log\norm{f_X(x)}_a^a.
\]
This in turn gives $\sigma_e(X)=\lim_{a\downarrow
  1}\norm{f_X(x)}_{a}^{2a/(1-a)}$. At $a=1+\epsilon$ ($\epsilon\rightarrow 0$), the
left side of \eqref{eq:young} gives $\sigma_e(X+Y)$. For a given $a$, there is one free parameter in the right hand side of
\eqref{eq:young}. By choosing the free parameter such that the right side of
\eqref{eq:young} is maximized, we obtain the EPI. 
The case $n>1$ is obtained by replacing
entropy with entropy rate (and using a version of \eqref{eq:young} in
$\Reals^n$ to find conditions of equality). The equality in \eqref{eq:epi} results from
the equality in \eqref{eq:young}.

The EPI, in some sense, is the derivative of the Young's inequality. 
\end{proof}

Several remarks are in order now.

\paragraph{Bound on conditional discrete entropy}{%
Let $A$ and $B$ be finite discrete sets (alphabets). Since not all elements
of $A+B$ are distinct, we have the sumset inequality
\begin{equation}
\mu(A+B)\leq \mu(A)\mu(B),  
\label{eq:sumset}
\end{equation}
where $\mu$ denotes set cardinality. This in turn gives
\begin{equation}
H(X+Y)\leq H(X)+H(Y), 
\label{eq:H(X+Y)}   
\end{equation}
where $X$ and $Y$ are independent discrete random variables taking values,
respectively, in alphabets $A$ and $B$, and $H$ is discrete entropy. For
uniform random variables \eqref{eq:H(X+Y)} is just the sumset inequality
\eqref{eq:sumset}; non-uniform distributions can (almost) be
converted to uniform distributions via the asymptotic equipartition theorem
\cite{cover2006}. The inequality \eqref{eq:H(X+Y)} reflects the fact that
the sum of independent discrete random variables typically does not tend to
a uniform random variable (maximum entropy). In fact, in a sense, $X+Y$ is
``less uniform'' than $X$ and $Y$. In sharp contrast, the (normalized) sum
of independent continuous random variables tends to a Gaussian random
variable (maximum entropy)---however, the increase in randomness is
measured in entropy power, not the entropy itself.

The inequality \eqref{eq:H(X+Y)} seems to indicate that as noise is added
along the optical fiber, the conditional entropy of the signal does not
increase. Two distinct pairs $(\vect q^n_1,\vect w^n_1)$ and $(\vect
q^n_2,\vect w^n_2)$ can have the same sum $\vect q^n_1+\vect w^n_1=\vect
q^n_2+\vect w^n_2$, making $\vect Q^n+\vect W^n$ potentially ``less
random'', so to speak. This is, however, true only in a discrete-state
model in which $\vect q^n$ is quantized in a finite set. It follows that,
the entropy bounds in this paper may not be valid in discrete-state models,
due to important differences between the differential and discrete
entropies. This difference stems from the properties of the cardinality
(volume) in discrete (continuous) sets. 
}

\paragraph{Growth of the effective variance in evolution}{%
For a Gaussian random variable with variance $\sigma^2$,
$\sigma_e^2(X)=\sigma^2$. Thus one may think of $\sigma_e^2(X)$ as the
effective variance of $X$ or the squared radius of the support of $X$ (hence the notation).

A family of fascinating metric inequalities analogous to \eqref{eq:epi}
exist in geometry and analysis, where the squared radius
\eqref{eq:volume-epi} is defined differently \cite{gardner2002brunn}.
Notably, in one of its facets, the Brunn-Minkowski inequality (BMI) for
compact regions $A,B\subset\Reals^n$ states 
\begin{IEEEeqnarray}{rCl}
  \mu^{\frac{1}{n}}(A+B)\geq \mu^{\frac{1}{n}}(A)+\mu^{\frac{1}{n}}(B),
\label{eq:bmi}
\end{IEEEeqnarray}
where $\mu$ is the Lebesgue measure (volume) and $A+B$ is the
Minkowski sum of $A$ and $B$. The BMI looks like the EPI
with $\sigma_e^2(X)\eqdef\mu^{\frac{1}{n}}(A)$. 
Let $\mathcal{A}_\epsilon^m$,  $\mathcal{B}_\epsilon^m$ and
$\mathcal{C}_\epsilon^m$ be, respectively, the $\epsilon$-typical sets of
random variables $X,Y\in\Reals^n$ and $Z=X+Y$.  From the concentration of
measure $\mu(\mathcal A^m_{\epsilon})\rightarrow 2^{mh(X)}$, or
$\mu^{\frac{1}{nm}}(\mathcal A_{\epsilon}^m)\rightarrow
2^{\frac{1}{n}h(X)}$, as $\epsilon\rightarrow 0$.  Applying the BMI to
$\mathcal{A}_\epsilon^m$ and $\mathcal{B}_\epsilon^m$, we obtain the EPI
with factor one in the exponent in \eqref{eq:volume-epi} instead of two.
The result is not the desired EPI inequality. This is because
$\mathcal{C}_\epsilon^m\neq \mathcal{A}_\epsilon^m+\mathcal{B}_\epsilon^m$.
In fact, again from the concentration of measure, $Z^m$ concentrates on a
smaller set $\mathcal{C}_\epsilon^m\subset\mathcal{A}_\epsilon^m +
\mathcal{B}_\epsilon^m$, \ie, $\mu(\mathcal{C}_\epsilon^m)<\mu(\mathcal{A}_\epsilon^m)\mu(\mathcal{B}_\epsilon^m)$.
To obtain the EPI from the BMI, and thus to give the EPI a geometric
meaning, we need a probabilistic version of the Minkowski sum, where the
volume is defined as the size of high probability sequences.  Define the
$\Omega$-restricted Minkowski sum of two sets $A,B\subset\Reals^n$
\begin{IEEEeqnarray}{rCl}
A +_{\Omega} B\eqdef\Big\{ a+b \:|\: (a,b)\in \Omega\subset A\times B\Bigr\}.
\label{eq:rest-sum}  
\end{IEEEeqnarray}
The restricted BMI states that, if $\mu(\Omega)\geq (1-\delta)\mu(A)\mu(B)$
for some $\delta>0$, then \eqref{eq:bmi} holds but with exponent $2/n$
\cite[Theorem 1.2, with large $n$]{szarek1996vrms}. Furthermore, the restricted BMI is sharp, regardless
of how close $\Omega$ is to $A\times B$, \ie, as $\delta\rightarrow 0$. That is to
say, even a small uncertainty in the size of $A\times B$ would increase the
exponent in the BMI by a factor of two. The inequality is best seen for Gaussian random variables where typical
sets can be imagined as spherical shells \cite{szarek1996vrms}. Applying the restricted BMI to
$\mathcal{A}_\epsilon^m$ and $\mathcal{B}_\epsilon^m$ with 
$\Omega=\{(a,b)\:|\:a\in \mathcal{A}_{\epsilon}^m, b\in \mathcal{B}_{\epsilon}^m,  a+b\in \mathcal{C}_{\epsilon}^m\}$, we successfully
obtain the EPI.

With the geometric interpretation of the BMI for typical sequences, the upper bound \eqref{eq:upper-bound} is
trivial. The output typical set $\mathcal A_{\epsilon}^m(\vect Q^n(z))$ is covered in the sphere 
$S_{2nm}(\vect q_{c_1}^n(z), \sqrt{m(\const{P}_0+\sigma^2(z))}$, centered at some $\vect q^n_{c_1}(z)$. For a particular input
sequence $\vect q^n(0)$, as the
typical set of the signal and noise are overlapped in the optical link, the
resulting region can be packed by a sphere $S_{2nm}(\vect q^n_{c_2}(z),
\sqrt{m\sigma^2(z)})$, centered at some $\vect q^n_{c_2}(z)$. The
capacity sphere-packing interpretation gives \eqref{eq:upper-bound}.

Inequalities in the family to which \eqref{eq:epi} and \eqref{eq:bmi} belong
appear intimately connected; however, it seems difficult to deduce them all
from one master inequality, due to important
differences among them. There is substantial work on this type of
inequality; see \cite{gardner2002brunn} and references in \cite{cover2006}.
}

\bibliographystyle{IEEEtran}


\end{document}